\documentclass[runningheads]{llncs}
\usepackage{amsmath}
\usepackage{amssymb}
\usepackage{algorithm}
\usepackage{algorithmic}
\usepackage{fullpage}
\date{}

\title{Some Combinatorial Problems on Halin Graphs} 
\author{M.Kavin ${^1}$ \and K.Keerthana ${^1}$ \and N.Sadagopan ${^2}$ \and Sangeetha.S ${^1}$ \and R.Vinothini ${^1}$} 
\institute{${^1}$ Department of Computer Science and Engineering, College of Engineering Guindy, Chennai, India.\\ ${^2}$ Indian Institute of Information Technology, Design and Manufacturing, Kancheepuram, India. \\
\email{sadagopan@iiitdm.ac.in}}
\begin{document}
\maketitle
\begin{abstract}
Let $T$ be a tree with no degree 2 vertices and $L(T)=\{l_1,\ldots,l_r\}, r \geq 2$ denote the set of leaves in $T$. An Halin graph $G$ is a graph obtained from $T$ such that $V(G)=V(T)$ and $E(G)=E(T) \cup \{\{l_i,l_{i+1}\} ~|~ 1 \leq i \leq r-1\} \cup \{l_1,l_r\}$.  In this paper, we investigate   combinatorial problems such as, testing whether a given graph is Halin or not, chromatic bounds, an algorithm to color Halin graphs with the minimum number of colors. Further, we present polynomial-time algorithms for testing and coloring problems.
\end{abstract}
\section{Introduction}
The study of special graph classes like bipartite graphs, chordal graphs, planar graphs have attracted the researchers both from theoretical and application perspective.  On the theoretical front, combinatorial problems like vertex cover, odd cycle transversal which are otherwise NP-complete in general graphs have polynomial-time algorithms restricted to chordal graphs \cite{chordal}. However, problems like Hamilton path, odd cycle transversal are still NP-complete even on special graph like planar graphs \cite{garey} which calls for identifying a non-trivial subclass of planar graphs where Hamilton path (cycle) is polynomial-time solvable.  One such subclass is the class of Halin graphs.  Halin graph is a graph obtained from a tree with no degree two vertices by joining all leaves with a cycle \cite{halin} and by construction Halin graphs are planar.  Bondy \cite{hamilton} showed that Hamilton cycle is polynomial-time solvable on Halin graphs.  Further, Barefoot \cite{hamil} improved his result and showed that every Halin graph is Hamilton-connected. Special graph class like Halin graphs play a vital role in understanding the inherent complexity of the problem.    Halin graphs are a good candidate graph class between the class of trees and the class of planar graphs.  The reason Halin graph is popular in the literature is due to the fact many combinatorial problems restricted to planar graphs are NP-complete \cite{garey} and on trees, they are polynomial-time solvable.  This leaves open the complexity of the problems in subclass of planar graphs, for example, Halin graphs.  In particular, finding a maximum leaf spanning tree in general and planar graphs are NP-complete whereas on Halin graph it is polynomial-time solvable \cite{mlst}.  \\
Although, chordality testing and planarity testing have received attention in the past, the related question, namely, Halin graph testing (whether a graph is Halin or not), to the best of our knowledge has not been addressed in the literature which we address in this paper.  
%For a graph $G$, the odd cycle transversal is a set $S \subset V(G)$ such that $G \setminus S$ is bipartite.  It was shown by Yannakakis et al. \cite{yanna} that this problem is NP-complete in general graphs and Lokashtanov et al. \cite{lok} have shown that this problem NP-complete on planar graphs.  This leaves open the complexity of odd cycle transversal on subclasses of planar graphs.  One such subclass is the class of Halin graphs.  In this paper, we show that on Halin graphs, odd cycle transversal (OCT) is polynomial-time solvable which is the second non-trivial special graph class where OCT is solvable in polynomial time.  The only other graph class known in the literature where OCT is polynomial-time solvable is the class of chordal graphs \cite{chordal}.\\
A proper coloring of a graph is an assignment of colors to the vertices such that the adjacent vertices receive different color.  It is well-known that all planar graphs can be colored with 4 colors.  In this paper, we present some structural observations on Halin graphs using which we characterize 3-colorable Halin graphs and 4-colorable Halin graphs. 
\subsection{Graph Preliminaries}
Notation and definitions are as per \cite{west,golu}.  Let $G =(V,E)$ be an undirected non weighted simple graph, where $V(G)$ is the set of vertices and $E(G) \subseteq \{\{u,v\}~|~ u,v \in V(G)$, $u \not= v \}$. 
For a vertex $v \in V(G)$, $N_G(v)=\{u ~|~ \{u,v\} \in E(G) \}$.  The degree of a vertex $v$ is the size of $N_G(v)$.  A vertex $v$ is universal if $N_G(v)=V(G) \setminus \{v\}$. A graph is acylic if contains no cycle.  A tree is a connected and an acyclic graph. Let $T$ be a tree with no degree 2 vertices and $L(T)=\{l_1,\ldots,l_r\}$ denote the set of leaves in $T$.  An Halin graph $G$ is a graph obtained from $T$ such that $V(G)=V(T)$ and $E(G)=E(T) \cup \{\{l_i,l_{i+1}\} ~|~ 1 \leq i \leq r-1\} \cup \{l_1,l_r\}$.  Informally, an Halin graph can be seen as a graph constructed from a tree rooted at a vertex  with no degree two vertices and joining all leaves with a cycle.  We refer this informal description as {\em tree-cycle} representation of Halin graph.  %A graph is bipartite if it contains no odd cycle.  An odd cycle transversal (OCT) is a set $S \subset V(G)$ such that the graph $G \setminus S$ obtained from $G$ by removing $S$ is bipartite.   
A proper coloring of graph is a vertex coloring in which adjacent vertices receive different color. A graph is $k$-colorable if it can be properly colored with $k$ colors. Let $C_n$ denote a cycle graph on $n$ vertices with $V(G)=\{v_1,\ldots,v_n\}$ and $E(G)=\{\{v_i,v_{i+1} ~|~ 1 \leq i \leq n-1\} \cup \{v_1,v_n\}$.  A wheel graph $W_{n+1}$ on $n+1$ vertices is constructed using $C_n$ and a universal vertex. i.e., $V(W_{n+1})=\{v_1,\ldots,v_{n+1}\}$ such that the set $\{v_1,\ldots, v_n\}$ induces a cycle and $N_G(v_{n+1})=\{v_1,\ldots, v_n\}$.  %For a graph $G$ and pair $\{x,y\} \subset V(G)$, the graph $G \cdot xy$ is graph obtained from $G$ by contracting the pair $\{x,y\}$ and by contraction we mean the merging of vertices $x$ and $y$ into a vertex $z$ such that $N_{G \cdot xy}(z)=N_G(u) \cup N_G(v)$.  If $\{x,y\}$ is an edge $e$ in $G$, then we denote the resulting graph on contraction as $G \cdot e$. 
\section{Halin Graph Testing}
In this section, we present a polynomial-time algorithm to test whether a given graph is Halin or not. \\ {\bf Correctness and Run-time Analysis:} If the input graph $G$ is Halin, then there exists a tree-cycle representation of $G$ and our algorithm looks for one such representation.  Since, the tree-cycle representation can have any element of $V(G)$ as its root, we run Breadth First Search from each vertex and identify for which vertex it yields the tree-cycle representation.  If for all elements of $V(G)$, we do not get to see the tree-cycle representation, then we output $G$ is not Halin.  For a graph with $n$ vertices and $m$ edges, BFS runs in $O(n+m)$ time.  Checking whether $G_v$ induces a cycle or not incurs $O(n)$ time as $|L(T)|=O(n)$.  Therefore, the overall running time of Halin graph testing is $O(n) \times O(n+m)=O(mn)$, polynomial in the input size.  \\
\begin{algorithm}
\caption{Halin Graph Testing}
\begin{algorithmic}[h]
\FOR{each vertex $v$ in $G$}
\STATE{$S_v=\phi$}
\STATE{Find the Breadth First Search (BFS) Tree $T$ starting at $v$ and let $L(T)$ denotes the set of leaves in $T$}
\FOR{each edge $e=\{x,y\}$ in $E(G)\setminus E(T)$}
\STATE{$S_v=S_v \cup \{x,y\}$ }
\ENDFOR
\STATE{$G_v$ be the graph such that $V(G_v)=S_v$ and $E(G_v)=E(G)\setminus E(T)$}
\IF{$S_v=L(T)$ and $G_v$ is an induced cycle} 
\STATE{Output $G$ is Halin and exit.} 
\ELSE
\STATE{Break and run BFS again from another vertex}
\ENDIF
\ENDFOR
\STATE{Output $G$ is not a Halin graph.} 
\end{algorithmic}
\end{algorithm}
\section{On Chromatic Bounds of Halin Graphs}
It is well known that planar graphs are 4-colorable and hence, Halin graphs are 4-colorable.  In this section, we present some structural observations on Halin graphs using which we characterize 3-colorable and 4-colorable planar graphs.  
\begin{theorem}
\label{thm-color}
Let $G$ be an Halin Graph.  $G$ does not contain $W_{n+1}$, where $n=2k+1, k \in \mathbb{N} $ if and only if $\chi(G)=3$.
\end{theorem}
\begin{proof}
{\em Necessity:} Consider the tree-cycle representation of $G$ with vertex $v$ being the root of the tree.  Since $G$ does not contain degree two vertices by definition, we observe that there exists a triangle in $G$.  i.e., there exists $\{l_i,l_{i+1}\}$ in $L(T)$ such that $l_i$ and $l_{i+1}$ have a common parent $p$.  It is easy to see that the set $\{p,l_i,l_{i+1}\}$ induces a triangle in $G$.  Clearly, any proper coloring of $G$ requires three colors to color $\{p,l_i,l_{i+1}\}$.  Therefore, $\chi(G) \geq 3$.   Further, we know from Theorem \ref{coloringtheorem} that any wheel-free Halin graph can be colored with at most 3 colors.  Therefore, $\chi(G)=3$. \\{\em Sufficiency:} We prove this by contradiction.  Suppose $G$ contains $W_{n+1}$, where $n=2k+1, k \in \mathbb{N}$.  Since $G$ contains $C_{2k+1}$, we need three colors to color $C_{2k+1}$ and to color $W_{n+1}$, we need one more color.  Therefore, $G$ is 4-colorable, contradicting the given fact that $G$ is 3-colorable.  This completes the sufficiency and hence, the theorem follows. \qed
\end{proof}
\begin{corollary}
Let $G$ be an Halin Graph which contains $W_{n+1}$, where $n=2k+1, k \in \mathbb{N}$ as an induced subgraph.  Then, $\chi(G)=4$.
\end{corollary}
\begin{proof}
The claim follows from Theorem \ref{thm-color} and the fact that planar graphs are 4-colorable. \qed
\end{proof}
We call an Halin graph not containing $W_{n+1}$, where $n=2k+1, k \in \mathbb{N}$ as induced subgraph as {\em wheel free Halin graph}.  We now present an algorithm which produces a 3-coloring of wheel free Halin graph.  Further, we show that the algorithm runs in polynomial time.  We use the set $\{c_1,c_2,c_3\}$ of colors with $c_1$ being the color with least index.  {\em candidate color set} is the set of colors available for a vertex $u$ at that point of iteration and our algorithm picks the least indexed color from the candidate color set to color $u$.  For a vertex $u$, $color(u)$ denotes the color assigned to $u$ and $parent(u)$ denotes the parent of $u$.
\begin{algorithm}
\caption{3-coloring of wheel free Halin Graphs}
\label{coloring}
\begin{algorithmic}[1]
\STATE{ {\bf Input:} A wheel free Halin graph $G$}
\STATE{{\bf Output:} 3-coloring of $G$}
\STATE{Consider the {\em tree-cycle} representation of $G$ and let $T$ be the tree part of tree-cycle representation rooted at the vertex $v$.  To get $T$, simply remove the edges in the cycle part.}
\STATE{Initially, all vertices are uncolored.}
\STATE{With respect to $v$, using $T$, identify the left most node (leaf) $p$ and the right most node (leaf) $q$}
\STATE{Assign $color(p)=c_3$ and $color(parent(q))=c_3$}
\STATE{For the vertices in the path $P_{vs}$, where $s=parent(p)$, assign the colors alternately from the set $\{c_1,c_2\}$}
\STATE{For the vertices in the path $P_{vt}$, where $t=parent(parent(q))$, assign the colors alternately from the set $\{c_1,c_2\}$}
\STATE{To color the remaining vertices, perform the Depth First Search starting from $v$, for each uncolored vertex $u$, maintain the {\em candidate color set} which is the set of available colors from $\{c_1,c_2,c_3\}$ to color $u$.  Note some of the elements in $N_G(u)$ may be colored because of the previous steps}
\STATE{From the {\em candidate color set}, pick the color with the least index and assign the same to $color(u)$}
\end{algorithmic}
\end{algorithm}
\begin{theorem}
\label{coloringtheorem}
Algorithm \ref{coloring} colors the input graph with 3 colors.
\end{theorem}
\begin{proof}
To prove our claim, we observe that as per our algorithm, when we encounter an uncolored vertex $u$, either  all vertices in $N_G(u)$ are colored or only a subset $S \subset N_G(u)$ is colored and in either case, it is colored with at most two colors from the set $\{c_1,c_2,c_3\}$.  This is true for the following reasons: during DFS, when we visit an uncolored vertex $u$, $parent(u)$ and one of its neighbours are colored and together they have used up at most two colors.  Therefore, $u$ can be colored with the least indexed color from the candidate color set.  However, for the rightmost leaf $q$, all three vertices in $N_G(q)$ are already colored.  The Step-6 of our algorithm ensures that they still use at most two colors.  Therefore, we are left with the third color to color $q$.  Hence, the theorem. \qed
\end{proof}
%\section{Odd-cycle Transversal in Halin Graphs}
%Odd-cycle transversal of a graph $S$, asks for a minimum set $S \subset V(G)$ such that $G \setminus S$ is 2-colorable. Also, $G\setminus S$ becomes bipartite.  In this section, we show that finding such $S$ is polynomial-time solvable.  

\end{document}